\documentclass[11pt]{article}
\usepackage[letterpaper, left=1.5in, right=1.5in, bottom=1.5in, top=1.5in]{geometry}

\usepackage{here}
\usepackage{hyperref}
\usepackage{authblk}
\usepackage{amssymb}
\usepackage{amsmath,amssymb}
\usepackage{amsthm}
\usepackage{tabls}
\usepackage[pdftex]{graphicx}
\usepackage{wrapfig}
\usepackage{array}
\usepackage{url}
\usepackage{multirow}
\usepackage{multicol}
\usepackage{here}
\usepackage{algorithm}
\usepackage{algorithmic}
\usepackage{setspace}
\usepackage{xspace}

\usepackage[final,mode=multiuser]{fixme}

\fxsetup{theme=color}
\definecolor{fxtarget}{rgb}{0.0000,0.0000,0.4823}

\fxuseenvlayout{colorsig}
\fxusetargetlayout{color}

\fxuseenvlayout{colorsig}
\fxusetargetlayout{color}
\FXRegisterAuthor{yy}{ayy}{YY}
\FXRegisterAuthor{yn}{ayn}{YN}
\FXRegisterAuthor{si}{asi}{SI}
\FXRegisterAuthor{hb}{ahb}{HB}
\fxsetface{env}{}
\newcommand{\algFt}{I\xspace}
\newcommand{\algSd}{II\xspace}
\newcommand{\algTd}{III\xspace}

\newtheorem{theorem}{Theorem}

\newtheorem{lemma}{Lemma}
\newtheorem{observation}{Observation}{\bfseries}{\itshape}

\theoremstyle{definition}

\newtheorem{example}{Example}

\newtheorem{problem}{Problem}

\newcommand{\lcs}{\mathsf{lcs}}
\newcommand{\dline}[1]{\langle #1 \rangle}
\newcommand{\Undef}{\mathsf{undefined}}
\newcommand{\beg}{\mathsf{b}}
\newcommand{\en}{\mathsf{e}}

\begin{document}

\title{Faster Space-Efficient STR-IC-LCS Computation}

\author[1]{Yuki~Yonemoto}
\author[2]{Yuto~Nakashima}
\author[2]{Shunsuke~Inenaga}
\author[3]{Hideo~Bannai}

\affil[1]{Department of Information Science and Technology, Kyushu University}
\affil[2]{Department of Informatics, Kyushu University}
\affil[3]{M\&D Data Science Center, Tokyo Medical and Dental University}

\date{}
\maketitle

\begin{abstract}
  One of the most fundamental method for comparing two given strings $A$ and $B$ 
  is the \emph{longest common subsequence}~(LCS), 
  where the task is to find (the length) of an
  \yynote*{modified on comment 1}{%
  LCS of $A$ and $B$.
  }%
  In this paper, we deal with the \emph{STR-IC-LCS}\footnote{STR-IC-LCS stands for ``subSTRing InCluding LCS~\cite{SEQECLCS_Chen_2011}.} problem 
  which is one of the constrained LCS problems proposed by Chen and Chao [J. Comb. Optim, 2011].
  A string $Z$ is said to be an STR-IC-LCS of three given strings $A$, $B$, and $P$,
  \yynote*{modified on comment 3}{%
  if $Z$ is a longest string satisfying that (1) $Z$ includes $P$ as a substring and (2) $Z$ is a common subsequence of $A$ and $B$.
  }%
  We present three efficient algorithms for this problem:
  First, we begin with a space-efficient solution
  which computes the length of an STR-IC-LCS in $O(n^2)$ time and $O((\ell+1)(n-\ell+1))$ space,
  where $\ell$ is the length of an 
  \yynote*{modified on comment 1}{%
  LCS
  }%
  of $A$ and $B$ of length $n$.
  When $\ell = O(1)$ or $n-\ell = O(1)$,
  then this algorithm uses only linear $O(n)$ space.
  Second, we present a faster algorithm that works in
  $O(nr/\log{r}+n(n-\ell+1))$ time, where $r$ is the length of $P$,
  while retaining the $O((\ell+1)(n-\ell+1))$ space efficiency.
  Third, we give an alternative algorithm that runs in
  $O(nr/\log{r}+n(n-\ell'+1))$ time with $O((\ell'+1)(n-\ell'+1))$ space,
  where $\ell'$ denotes the STR-IC-LCS length for input strings $A$, $B$, and $P$.
\end{abstract}


%
\section{Introduction} 

Comparison of two given strings (sequences) has been a central task in Theoretical Computer Science, since it has many applications including alignments of biological sequences, spelling corrections, and similarity searches.

One of the most fundamental method for comparing two given strings $A$ and $B$ is the \emph{longest common subsequence} \emph{LCS}, where the task is to find (the length of) a common subsequence $L$ that can be obtained by removing zero or more characters from both $A$ and $B$, and no such common subsequence longer than $L$ exists.
A classical dynamic programming (DP) algorithm is able to compute an LCS of $A$ and $B$ in quadratic $O(n^2)$ time with $O(n^2)$ working space, where $n$ is the length of the input strings~\cite{Wagner_1974_LCS}.
In the word RAM model with $\omega$ machine word size,
the so-called ``Four-Russian'' method allows one to compute the length of an LCS of two given strings in $O(n^2 / k + n)$ time, for any $k \leq \omega$, in the case of constant-size alphabets~\cite{MasekP80}.
Under a common assumption that $\omega = \log_2 n$,
this method leads to weakly sub-quadratic 
\yynote*{modified on comment 4}{%
$O(n^2 / \log{n})$
}%
time solution for constant alphabets.
In the case of general alphabets, the state-of-the-art algorithm computes
the length of an LCS in $O(n^2 \log^2 k / k^2 + n)$ time~\cite{BilleF08},
which is weakly sub-quadratic $O(n^2 (\log \log n)^2 / \log ^2 n)$ time for $k \leq \omega = \log_2 n$.
It is widely believed that such ``log-shaving'' improvements would be the best possible one can hope, since an $O(n^{2-\epsilon})$-time LCS computation for any constant $\epsilon > 0$ refutes the famous strong exponential time hypothesis (SETH)~\cite{AbboudBW15}.

Recall however that this conditional lower-bound under the SETH does not enforce us to use (strongly) quadratic \emph{space} in LCS computation.
Indeed, a simple modification to the DP method permits us to compute the length of an LCS in $O(n^2)$ time with $O(n)$ working space.
There also exists an algorithm that computes an LCS string in
$O(n^2)$ time with only $O(n)$ working space~\cite{Hirschberg75}.
The aforementioned log-shaving methods~\cite{MasekP80,BilleF08} use only $O(2^k + n)$ space, which is $O(n)$ for $k \leq \omega = \log_2 n$.

In this paper, we follow a line of research called the \emph{Constrained} LCS problems,
in which a pattern $P$ that represents a-priori knowledge of a user is given as a third input, and the task is to compute the longest common subsequence of $A$ and $B$ that meets the condition w.r.t. $P$~\cite{CLCS_Tsai_2003,SEQICLCS_2004,SEQECLCS_Chen_2011,STRICLCS_DEOROWICZ_2012,SEQ-IC-LCS-RLE_2014,STR-IC-LCS-RLE_Kuboi_2017}.
The variant we consider here is the \emph{STR-IC-LCS} problem of computing
a longest string $Z$ which satisfies that
(1) $Z$ \emph{includes} $P$ as a \emph{substring} and
(2) $Z$ is a common subsequence of $A$ and $B$.
We present new solutions to the STR-IC-LCS problem, named
Algorithm~\algFt, Algorithm~\algSd, and Algorithm~\algTd.

Algorithm~\algFt is the first space-efficient algorithm for the STR-IC-LCS problem
running in $O(n^2)$ time with
  $O((\ell+1)(n-\ell+1))$ working space,
where $\ell = \lcs(A,B)$ denotes the length of an LCS of $A$ and $B$.
This solution improves on the state-of-the-art STR-IC-LCS algorithm
of Deorowicz~\cite{STRICLCS_DEOROWICZ_2012} that uses $\Theta(n^2)$ time and $\Theta(n^2)$ working space, since $O((\ell+1)(n-\ell+1)) \subseteq O(n^2)$ always holds.
This method requires only sub-quadratic $o(n^2)$ space whenever $\ell = o(n)$.
In particular,
  when $\ell = O(1)$ or $n-\ell = O(1)$,
  which can happen when we compare very different strings
  or very similar strings, respectively,
then our algorithm uses only linear $O(n)$ space.
Algorithm~\algFt, is built on a non-trivial extension of the LCS computation algorithm by Nakatsu et al.~\cite{DBLP:journals/acta/NakatsuKY82} that runs in $O(n(n-\ell+1))$ time with $O((\ell+1)(n-\ell+1))$ working space (Section~\ref{sec:solution}).

Algorithm~\algSd is a faster version of Algorithm~\algFt,
which works in faster $O(nr/\log{r}+n(n-\ell+1))$ time with the same $O((\ell+1)(n-\ell+1))$ working space, where $r = |P|$ (Section~\ref{sec:f_solution}). Recall that $r \leq n$ holds.

Algorithm~\algTd is an alternative version to Algorithm~\algSd,
that works in $O(nr/\log{r}+n(n-\ell'+1))$ time with $O((\ell'+1)(n-\ell'+1))$ space, where $\ell'$ is the length of an STR-IC-LCS of $A$, $B$, and $P$. Since $\ell' \leq \ell$, $n-\ell'+1 \geq n - \ell + 1$ holds, implying that Algorithm~\algTd takes at least as much time as Algorithm~\algSd.
Still, we show that Algorithm~\algTd uses less space than Algorithm~\algSd for some strings, by presenting strings for which
$\ell' = O(1)$ and $\ell = \Theta(n)$ (Section~\ref{sec:alternative_solution}).

Table~\ref{tab:relatedwork} summarizes the complexities of the existing and proposed algorithms for STR-IC-LCS.

We remark that the $O(n^{2-\epsilon})$-time conditional lower-bound for LCS
also applies to our case since STR-IC-LCS with the pattern $P$ being the empty string is equal to LCS, and thus, our solution is almost time optimal.

\begin{table}[hbtp]
  \caption{Time and space complexities of algorithms for STR-IC-LCS, 
  for input strings $A$ and $B$ of length $n$ and constraint string $P$ of length $r$, where $r \leq n$. $\ell$ denotes the LCS length of $A$ and $B$, and $\ell'$ denotes the STR-IC-LCS length of $A$, $B$, and $P$.}
  \label{tab:relatedwork}
  \centering{
  \renewcommand{\arraystretch}{1.3}
  \begin{tabular}{|cc|c|c|}\hline
    \multicolumn{2}{|c|}{Algorithm} & Time Complexity & Space Complexity\\\hline\hline
    \multicolumn{2}{|c|}{Deorowicz's Algorithm~\cite{STRICLCS_DEOROWICZ_2012}} & $O(n^2)$ & $O(n^2)$\\
    Algorithm \algFt & \multirow{3}{*}{[our work]} & $O(n^2)$ & $O((\ell+1)(n-\ell+1))$\\
    Algorithm \algSd & & $O(nr/\log{r}+n(n-\ell+1))$ & $O((\ell+1)(n-\ell+1))$\\
    Algorithm \algTd & & $O(nr/\log{r}+n(n-\ell'+1))$ & $O((\ell'+1)(n-\ell'+1))$\\\hline\hline
  \end{tabular}
  }
\end{table}

A preliminary version of this work appeared in~\cite{YonemotoNIB23},
in which Algorithm \algFt was proposed.
The new materials in this full version are our second and third solutions, Algorithm~\algSd and Algorithm~\algTd, described in Sections~\ref{sec:f_solution} and \ref{sec:alternative_solution}.

\subsection*{Related Work}
There exist four variants of the \emph{Constrained} LCS problems,
STR-IC-LCS/SEQ-IC-LCS/STR-EC-LCS/SEQ-EC-LCS,
each of which is to compute a longest string $Z$ such that
(1) $Z$ includes/excludes the constraint pattern $P$ as a substring/subsequence and 
(2) $Z$ is a common subsequence of the two target strings $A$ and $B$~\cite{CLCS_Tsai_2003,SEQICLCS_2004,SEQECLCS_Chen_2011,STRICLCS_DEOROWICZ_2012,SEQ-IC-LCS-RLE_2014,STR-IC-LCS-RLE_Kuboi_2017}.
Yamada et al.~\cite{YamadaNIBT20} proposed
an $O(n\sigma + (\ell''+1)(n-\ell''+1)r)$-time and space algorithm
for the STR-EC-LCS problem, which is also based on the method by Nakatsu et al.~\cite{DBLP:journals/acta/NakatsuKY82}, where $\sigma$ is the alphabet size, $\ell''$ is the length of an STR-EC-LCS and $r$ is the length of $P$.
However, the design of our solution to STR-IC-LCS is quite different from
that of Yamada et al.'s solution to STR-EC-LCS.

\section{Preliminaries}\label{sec:preliminaries}

\subsection{Strings}
Let $\Sigma$ be an {\em alphabet}.
An element of $\Sigma^*$ is called a {\em string}.
The length of a string $S$ is denoted by $|S|$.
The empty string $\varepsilon$ is a string of length 0.
For a string $S = uvw$, $u$, $v$ and $w$ are called
a \emph{prefix}, \emph{substring}, and \emph{suffix} of $S$, respectively.

The $i$-th character of a string $S$ is denoted by $S[i]$, where $1 \leq i \leq |S|$.
For a string $S$ and two integers $1 \leq i \leq j \leq |S|$,
let $S[i..j]$ denote the substring of $S$ that begins at position $i$ and ends at
position $j$, namely, $S[i..j] = S[i] \cdots S[j]$.
For convenience, let $S[i..j] = \varepsilon$ when $i > j$.
$S^R$ denotes the reversed string of $S$, i.e., $S^R = S[|S|] \cdots S[1]$.
A non-empty string $Z$ is called a \emph{subsequence} of another string $S$
if there exist increasing positions $1 \leq i_1 < \cdots < i_{|Z|} \leq |S|$
in $S$ such that $Z = S[i_1] \cdots S[i_{|Z|}]$.
The empty string $\varepsilon$ is a subsequence of any string.
A string that is a subsequence of two strings $A$ and $B$
is called a \emph{common subsequence} of $A$ and $B$.

\subsection{STR-IC-LCS}
Let $A, B$, and $P$ be strings.
A string $Z$ is said to be {\em an STR-IC-LCS} of two target strings $A$ and $B$ {\em including} the pattern $P$ if $Z$ is a longest string such that
(1) $P$ is a substring of $Z$ and
(2) $Z$ is a common subsequence of $A$ and $B$.

For ease of exposition, we assume that $n = |A| = |B|$,
but our algorithm to follow can deal with the general case where $|A| \neq |B|$.
We can also assume that $|P| \leq n$,
since otherwise there clearly is no solution.
In this paper, we present a space-efficient algorithm that computes
an STR-IC-LCS in $O(n^2)$ time and $O((\ell+1)(n-\ell+1))$ space,
where $\ell = \lcs(A, B)$ is the longest common subsequence length of $A$ and $B$.
In case where there is no solution, we use a convention that $Z = \bot$ and its length $|\bot|$ is $-1$. We remark that $\ell \geq |Z|$ always holds.

\section{Space-efficient solution (Algorithm \algFt) for STR-IC-LCS problem} \label{sec:solution}

In this section, we propose a space-efficient solution for the STR-IC-LCS problem.

\begin{problem} [STR-IC-LCS problem]
For any given strings $A, B$ of length $n$ and $P$,
compute an STR-IC-LCS of $A, B$, and $P$.
\end{problem}

\begin{theorem} \label{thm:str-ic-lcs}
    The STR-IC-LCS problem can be solved in $O(n^2)$ time and $O((\ell+1)(n-\ell+1))$ space
    where $\ell$ is the length of LCS of $A$ and $B$.
\end{theorem}
In Section~\ref{subsec:overview}, we explain an overview of our algorithm.
In Section~\ref{subsec:prefix-lcs}, we show a central technique for our space-efficient solution
and Section~\ref{subsec:algorithm} concludes with the detailed algorithm.

\subsection{Overview of our solution} \label{subsec:overview}

Our algorithm is built on the previous algorithm for the STR-IC-LCS problem
which was proposed by Deorowicz~\cite{STRICLCS_DEOROWICZ_2012}.
Firstly, we explain an outline of his algorithm.
\sinote*{changed for comment 6}{%
  An interval $[i..j]$ over the string $A$ is said to be
  a \emph{minimal occurrence} of string $P$,
  if $P$ is a subsequence of $A[i..j]$ and
  $P$ is not a subsequence of either $A[i+1..j]$ or $A[i..j-1]$.
  In what follows, we will simply call such $[i..j]$
  as a \emph{minimal interval} for an arbitrarily fixed $P$.
}%
Let $I_A$ be the set of minimal intervals over $A$,
\sinote*{changed for comment 7}{%
  whose size is defined to be the number of intervals in it.
}%
Remark that $I_A$ is of size linear in the length of $A$
since each interval in $I_A$ cannot contain any other intervals in $I_A$.
There exists a pair of minimal intervals $[\beg_A..\en_A]$ over $A$ and $[\beg_B..\en_B]$ over $B$
such that the length of an STR-IC-LCS is equal to the sum of
the three values $\lcs(A[1..\beg_A-1], B[1..\beg_B-1])$, $|P|$, and $\lcs(A[\en_A+1..n], B[\en_B+1..n])$
(see also Figure~\ref{fig:striclcs} for an example).
\begin{figure}[t]
    \centerline{\includegraphics[width=0.5\linewidth]{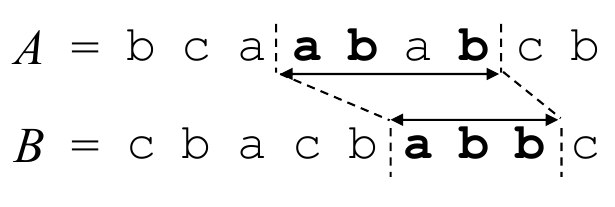}}
    \caption{
        Let $A = \mathtt{bcdababcb}$, $B = \mathtt{cbacbaaba}$, and $P = \mathtt{abb}$.
        The length of an STR-IC-LCS of these strings is 6.
        One of such strings can be obtained by minimal intervals $[4..7]$ over $A$ and $[6..8]$ over $B$
        because $\lcs(\mathtt{bca},\mathtt{cbacb}) = 2$, $|P| = 3$, and $\lcs(\mathtt{cb},\mathtt{c}) = 1$.
    }\label{fig:striclcs}
\end{figure}

First, the algorithm computes $I_A$ and $I_B$ and computes the sum of three values for any pair of intervals.
If we have an LCS table $d$ of size $n \times n$ such that $d(i, j)$ stores $\lcs(A[1..i], B[1..j])$ for any integers $i, j \in [1..n]$,
we can check any LCS value between prefixes of $A$ and $B$ in constant time.
It is known that this table can be computed in $O(n^2)$ time by using a simple dynamic programming.
Since the LCS tables for prefixes and suffixes require $O(n^2)$ space, the algorithm also requires $O(n^2)$ space.

Our algorithm uses a space-efficient LCS table by Nakatsu et al.~\cite{DBLP:journals/acta/NakatsuKY82}
instead of the table $d$ for computing LCSs of prefixes (suffixes) of $A$ and $B$.
The algorithm by Nakatsu et al. also computes a table by dynamic programming,
but the LCS values of $\lcs(A[1..i], B[1..j])$ for some pairs $(i, j)$
are missing in their table.
In what follows, we show how we can resolve this issue.

\subsection{Space-efficient prefix LCS} \label{subsec:prefix-lcs}
First, we explain a dynamic programming solution by Nakatsu et al. for computing an LCS of given strings $A$ and $B$.
We give a slightly modified description in order to describe our algorithm.
For any integers $i, s \in [1..n]$, let $f_{A}(s, i)$ be the length of the shortest prefix
$B[1..f_{A}(s, i)]$ of $B$
such that the length of the longest common subsequence of $A[1..i]$ and
$B[1..f_{A}(s, i)]$ is $s$.
For convenience, $f_{A}(s, i) = \infty$ if no such prefix exists. 
The values $f_{A}(s, i)$ will be computed using dynamic programming as follows: 
\[
    f_{A}(s, i) = \min\{f_{A}(s,i-1), j_{s, i}\},
\]
where $j_{s,i}$ is the index of the leftmost occurrence of $A[i]$ in $B[f_A(s-1,i-1)+1..n]$.
Let $s'$ be the largest value such that $f_{A}(s', i) < \infty$ for some $i$,
i.e, the $s'$-th row is the lowest row which has an integer value in the table $f_A$.
We can see that the length of the longest common subsequence of $A$ and $B$ is $s'$ (i.e., $\ell = \lcs(A, B) = s'$).
See Figure~\ref{fig:whole-tables} for an instance of $f_A$.
\begin{figure}[t]
     \centerline{\includegraphics[width=1.0\linewidth]{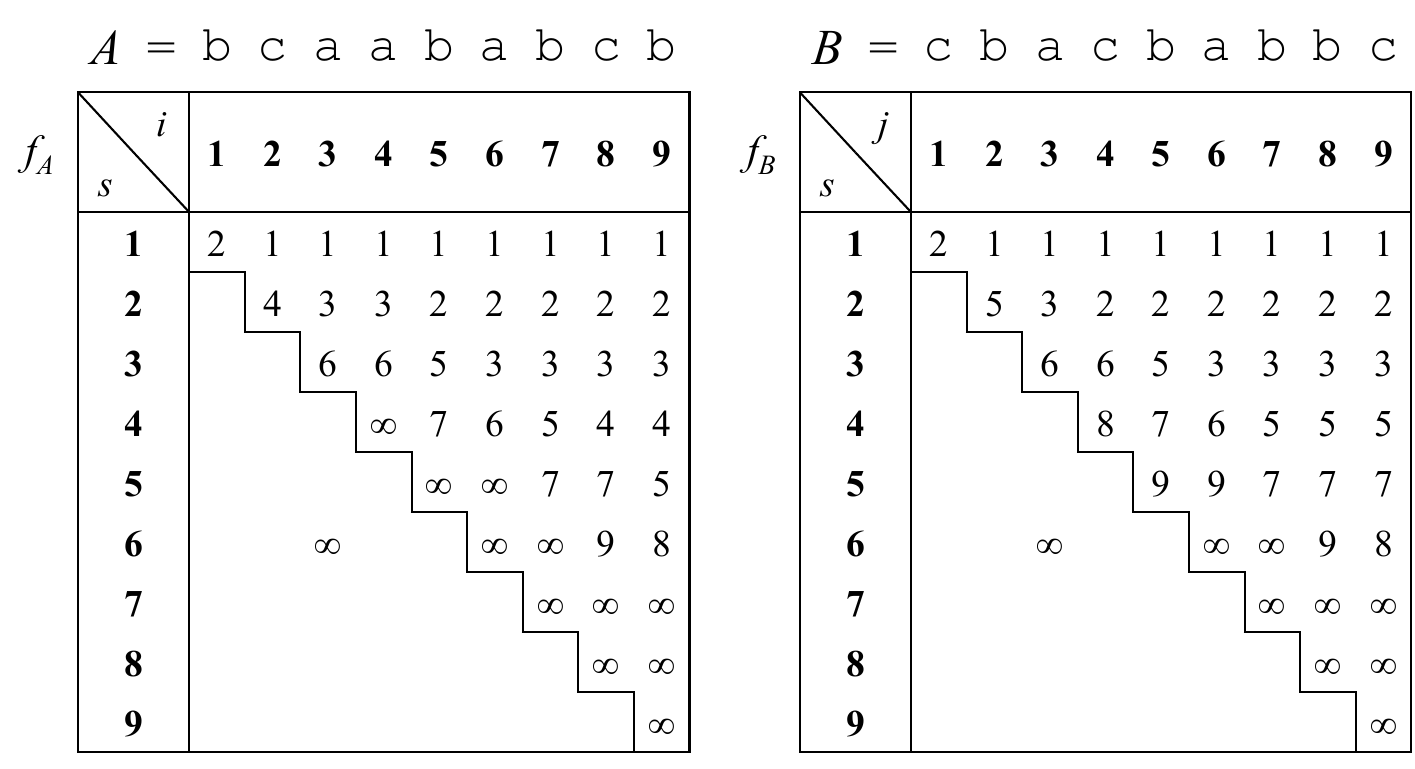}}
     \caption{
         The LCS-table $f_A$ which is defined by Nakatsu et al. of $A = \mathtt{bcdababcb}$.
         This figure also illustrates the table $f_B$ of $B = \mathtt{cbacbaaba}$.
     }\label{fig:whole-tables}
 \end{figure}

Due to the algorithm, we do not need to compute all the values in the table $f_A$ for obtaining the length of an LCS.
Let $F_A$ be the sub-table of $f_A$ such that $F_A(s, i)$ stores a value $f_A(s, i)$
if $f_A(s, i)$ is computed in the algorithm of Nakatsu et al.
Intuitively, $F_A$ stores the first $n-l+1$ diagonals of length at most $l$.
Let $\dline{i}$ be the set of pairs in the $i$-th diagonal line ($1 \leq i \leq n$) of the table $f_A$:
\[
    \dline{i} = \{ (s, i+s-1) \mid 1 \leq s \leq n-i+1 \}.
\]
Formally, $F_A(s, i) = \Undef$ if 
\begin{enumerate}
    \item $s>i$,
    \item $(s, i) \in \dline{j}~(j > n-\ell+1)$, or
    \item $F_A(s-1, i-1) \in \{\infty, \Undef\}$.
\end{enumerate}
Any other $F_A(s, i)$ stores the value $f_A(s, i)$.
Since the lowest row number of each diagonal line $\dline{j}~(j > n-\ell+1)$ is less than $\ell$,
we do not need to compute values which is described by the second item.
Actually, we do not need to compute the values in $\dline{n-\ell+1}$ for computing the LCS
since the maximum row number in the last diagonal line is also $\ell$.
However, we need the values on the last line in our algorithm.
Hence the table $F_A$ uses $O((\ell+1)(n-\ell+1))$ space 
(subtable which need to compute is parallelogram-shaped of height $\ell$ and base $n-\ell$).
See Figure~\ref{fig:sparse-tables} for an instance of $F_A$.

\begin{figure}[h]
     \centerline{\includegraphics[width=1.0\linewidth]{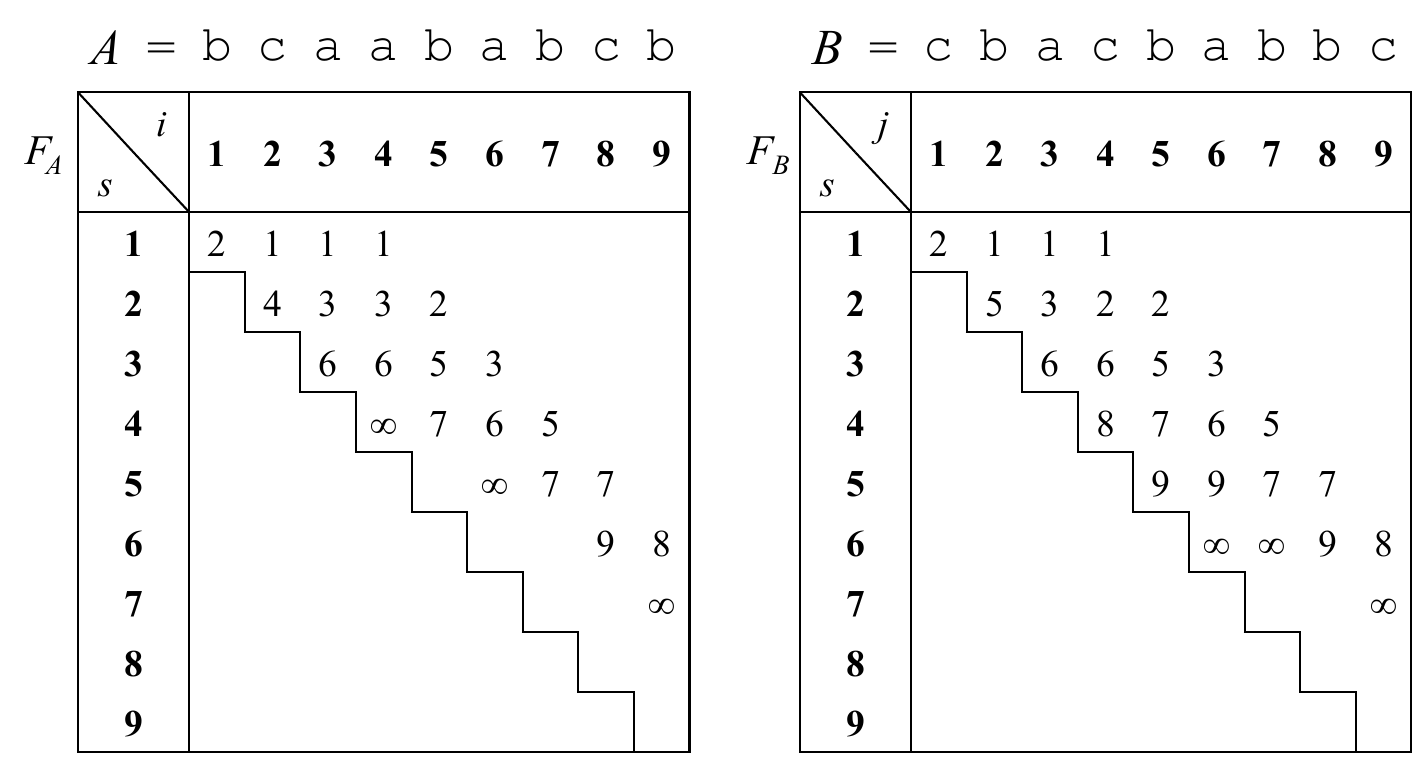}}
     \caption{
         A sparse table $F_A$ of $f_A$ for $A = \mathtt{bcdababcb}$ and $B = \mathtt{cbacbaaba}$
         does not give $\lcs(A[1..i], B[1..j])$ for some $(i, j)$.}
     \label{fig:sparse-tables}
\end{figure}


Now we describe a main part of our algorithm.
Recall that a basic idea is to compute $\lcs(A[1..i], B[1..j])$ from $F_A$.
If we have all the values on the table $f_A$, we can check the length $\lcs(A[1..i], B[1..j])$ as follows.
\begin{observation} \label{obs:lcs}
    The length of an LCS of $A[1..i]$ and $B[1..j]$ for any $i, j \in [1..n]$
    is the largest $s$ such that $f_A(s, i) \leq j$.
    If no such $s$ exists, $A[1..i]$ and $B[1..j]$ have no common subsequence of length $s$.
\end{observation}
However, $F_A$ does not store several integer values with respect to the second condition of $\Undef$ for some $i$ and $j$.
See also Figure~\ref{fig:sparse-tables} for an example of this fact.
In this example, we can see that $\lcs(A[1..7], B[1..4]) = f_A(3, 7) = 3$ from the table $f_A$,
but $F_A(3, 7) = \Undef$ in $F_A$.
In order to resolve this problem, we also define $F_B$ (and $f_B$).
Formally, for any integers $j, s \in [1..n]$, let $f_{B}(s, j)$ be the length of the shortest prefix
$A[1..f_{B}(s, j)]$ of $A$
such that the length of the longest common subsequence of $B[1..j]$ and
$A[1..f_{B}(s, j)]$ is $s$.
Our algorithm accesses the length of an LCS of $A[1..i]$ and $B[1..j]$ for any given $i$ and $j$ by using two tables $F_A$ and $F_B$.
The following lemma shows a key property for the solution.
\begin{lemma} \label{lem:recover-lcs}
    Let $s$ be the length of an LCS of $A[1..i]$ and $B[1..j]$.
    If $F_A(s, i) = \Undef$, then $F_B(s, j) \neq \Undef$.
\end{lemma}
This lemma implies that the length of an LCS of $A[1..i]$ and $B[1..j]$ can be obtained if we have the two sparse tables (see also Figure~\ref{fig:lack-value}).
\begin{figure}[H]
    \centerline{\includegraphics[width=1.0\linewidth]{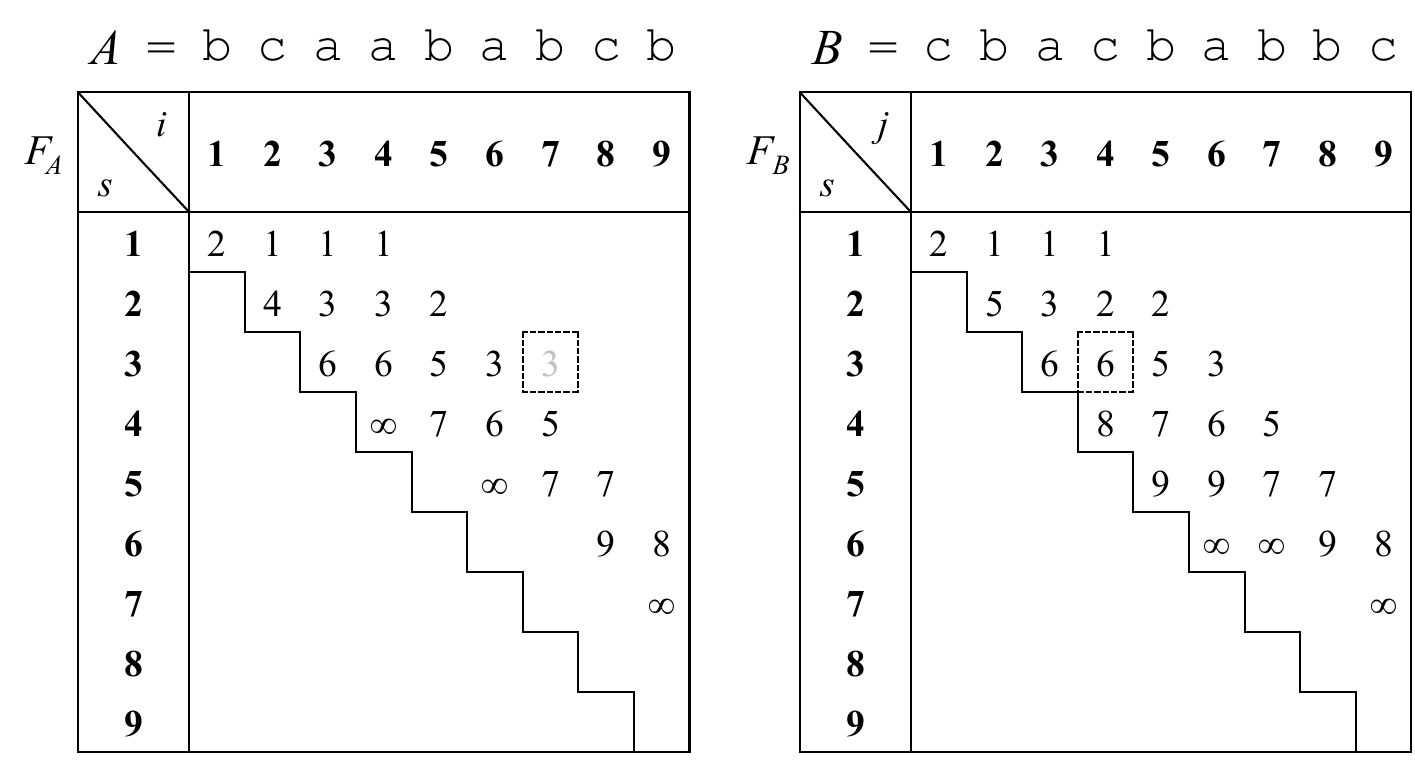}}
    \caption{
        Due to Observation~\ref{obs:lcs}, $f_A(3,7)$ gives the fact that $\lcs(A[1..7],B[1..4]) = 3$.
        However, $F_A(3,7) = \Undef$.
        Then we can obtain the fact that $\lcs(A[1..7],B[1..4]) = 3$ by using $F_B$.
        Namely, $F_B(3,4)$ gives the LCS value.
    }\label{fig:lack-value}
\end{figure}

Before proving Lemma~\ref{lem:recover-lcs},
we show Lemma~\ref{lem:visible-lcs} below:
\begin{lemma} \label{lem:visible-lcs}
    Let $U_{F_A}$ be the set of pairs $(s, i)$ of integers
such that $F_B(s, j) \neq \Undef$, where $F_A(s, i) = j$.
    For any $1 \leq s \leq \lcs(A, B)$,
    there exists $i$ such that $(s, i) \in U_{F_A}$,
\end{lemma}
\begin{proof}
    Let $\ell = \lcs(A,B)$.
    \sinote*{changed}{%
    We consider a sequence $i_1, \ldots, i_s, \ldots, i_{\ell}$
of positions in $A$ such that $A[i_1] \cdots A[i_s] \cdots A[i_{\ell}]$ is an LCS of $A$ and $B$ which can be obtained by backtracking over $F_A$.
    }%
    Suppose that $F_B(s, F_A(s, i_s)) = \Undef$ for some $s \in [1..\ell]$.
    \yynote*{added on comment 9}{%
    $F_B(s', F_A(s', i_{s'})) = \Undef$ for any $s' \in [s..\ell]$, 
    since $F_A(1,i_1) < \ldots < F_A(\ell,i_{\ell})$ and $s'$ increases as one goes either to the next diagonal cell or even further right.
    }%
    However, $F_B(\ell,F_A(\ell,i_{\ell}))$ is not $\Undef$.
    Therefore, $F_B(s, F_A(s, i_s)) \neq \Undef$ for any $s \in [1..\ell]$.
    Thus that the lemma holds.
\end{proof}
Now we are ready to prove Lemma~\ref{lem:recover-lcs} as follows.
\begin{proof}[Proof of Lemma~\ref{lem:recover-lcs}]
    \sinote*{modified}{%
    Let $\ell = \lcs(A,B)$ and $X = A[i_1] \cdots A[i_s] \cdots A[i_{\ell}]$ be an LCS of $A$ and $B$ which can be obtained from $F_A$, as in the proof for Lemma~\ref{lem:visible-lcs}.
    We also consider a sequence $j_1, \ldots, j_s, \ldots, j_{\ell}$
of positions in $B$ such that $B[i_1] \cdots B[j_s] \cdots B[j_{\ell}]$ is an LCS of $A$ and $B$, which satisfies that $j_s = F_A(k,i_s)$ for each $s \in [1..\ell]$.
    }%

    Assume that $F_A(s,i) = \Undef$.
    Let $m$ be the largest integer such that $i_{s+m} \leq i$ holds 
    \yynote*{added on comment 13}{%
    and $F_A(s+m,i_{s+m}) \neq \Undef$.
    }%
    If no such $m$ exists, namely $i < i_1$, the statement holds since $F_A(s,i) \neq \Undef$.
    Due to 
    \yynote*{modified on comment 11}{%
    Observation~\ref{obs:lcs}, 
    }%
    $F_A(s+m,i) > j$.
    Thus $j < j_{s+m}$ holds.
    On the other hand, we consider the table $F_B$ (and $f_B$).
    Let $i' = f_B(s,j)$ and $i'' = f_B(s+m,j)$.
    Due to Observation~\ref{obs:lcs}, $i' \leq i < i''$ holds.
    By Lemma~\ref{lem:visible-lcs}, $F_B(s+m,j_{s+m}) \neq \Undef$.
    This implies that $F_B(s+m,j)~(= i'')$ is not $\Undef$.
    By the definition of $X$, $j_{s+m}-j \geq m-1$.
    Notice that $(s,j)$ is in $(j-s+1)$-th diagonal line.
    These facts imply that $F_B(s,j) \neq \Undef$.
    See also Figure~\ref{fig:proof-sketch} for an illustration.
\end{proof}
\begin{figure}[H]
    \centerline{\includegraphics[width=1.0\linewidth]{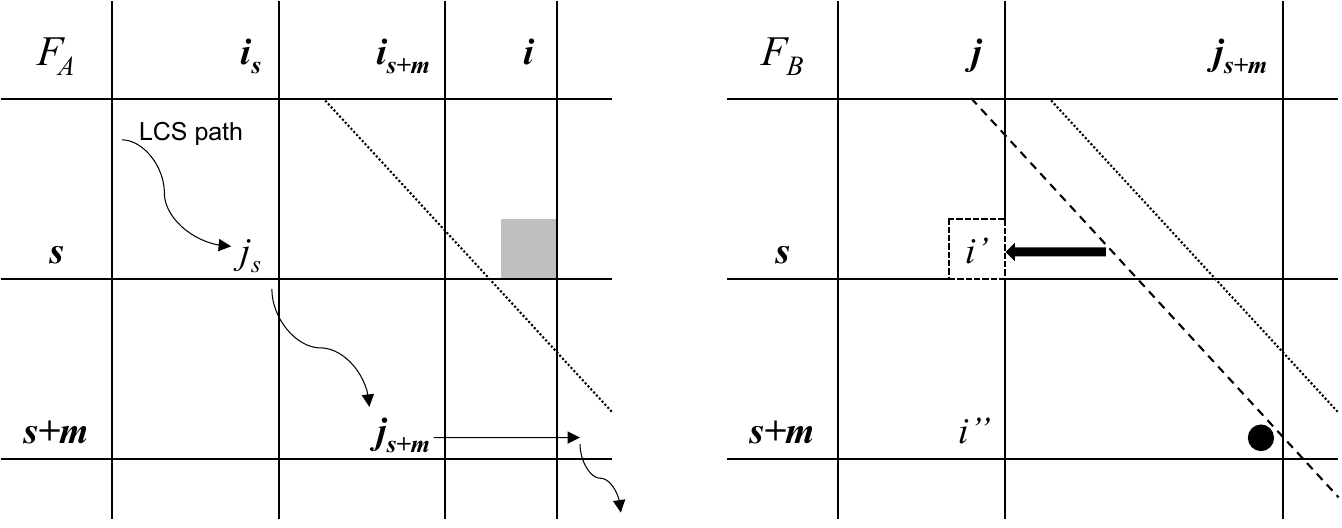}}
    \caption{
    This figure shows an illustration for the proof of Lemma~\ref{lem:recover-lcs} (and Lemma~\ref{lem:visible-lcs}).
    The length $s$ of an LCS of $A[1..i]$ and $B[1..j]$ cannot be obtained over $F_A$ because $F_A(s,i) = \Undef$ (the highlighted cell).
    However, the length can be obtained by $F_B(s,j)$ over $F_B$.
    The existence of $F_B(s+m,j_{s+m})$ from an LCS path guarantees the fact that $F_B(s,j) \neq \Undef$.
    }\label{fig:proof-sketch}
\end{figure}

\subsection{Algorithm \algFt} \label{subsec:algorithm}

A pseudo-code of our first space-efficient algorithm is given in Algorithm~\ref{alg:STR-IC-LCS}.

\begin{algorithm}[b!]
    \begin{spacing}{0.9}
        \caption{Algorithm~\algFt for computing the length of STR-IC-LCS}
        \label{alg:STR-IC-LCS}
        \begin{algorithmic}[1]
          \REQUIRE $A,B,P$\hspace{0.5em}($|A|=n,|B|=n,|P|=r$)
          \ENSURE $l,C$
          \vspace{1ex}
          \STATE compute $I_A$ and $I_B$
          \vspace{1ex}
          \STATE compute $F_A, F_B, F_{A^R}$, and $F_{B^R}$
          \STATE $\ell \leftarrow \lcs(A,B)$
          \STATE $l \leftarrow 0$
          \FOR {$i = 1$ to $|I_A|$} 
          \STATE $k^A_1 \leftarrow 1,~k^B_1 \leftarrow 1,~k^A_2 \leftarrow \ell,~k^B_2 \leftarrow \ell$
          \FOR {$j = 1$ to $|I_B|$}
          \STATE $k_1 \leftarrow 0,~k_2 \leftarrow 0$
          \STATE compute $\lcs(A[1..\beg_A(i)-1], B[1..\beg_B(j)-1])$ // as $k_1$ by Algorithm~\ref{alg:Serch4LCS_1}
          \STATE compute $\lcs(A[\en_A(i)+1..n], B[\en_B(j)+1..n])$ // as $k_2$ by Algorithm~\ref{alg:Serch4LCS_2}
          \IF {$k_1 + k_2 + r > l$}
          \STATE $l  \leftarrow k_1 + k_2 + r$
          \ENDIF
          \ENDFOR
          \ENDFOR
          \RETURN $l$ 
        \end{algorithmic}
        \vspace{1ex}
    \end{spacing}
  \end{algorithm}

  \begin{algorithm}[tbh]
    \begin{spacing}{0.9}
        \caption{Computing $\lcs(A[1..\beg_A(i)-1],B[1..\beg_B(j)-1])$ in Algorithm~\algFt}
        \label{alg:Serch4LCS_1}
        \begin{algorithmic}[1]
            \FOR {$k \leftarrow k^A_1$ to $\ell$}
                \IF {$F_A(\beg_A(i)-1, k) \leq \beg_B(j)-1$}
                    \IF {$F_A(\beg_A(i)-1, k+1) > \beg_B(j)-1$}
                        \STATE $k_1 \leftarrow k$
                        \STATE $k^A_1 \leftarrow k$
                        \STATE $\mathbf{break}$
                    \ENDIF
                \ELSIF {$F_A(\beg_A(i)-1, k) > \beg_B(j)-1$}
                    \IF {$F_A(\beg_A(i)-1, k-1) = \Undef$}
                        \STATE $k^A_1 \leftarrow k$
                        \FOR {$k' = k^B_1$ to $\ell$}                
                            \IF {$F_B(\beg_B(j)-1,k') > \beg_A(i)-1$}
                                \STATE $k_1 \leftarrow 0$
                                \STATE $k^B_1 \leftarrow k'$
                                \STATE $\mathbf{break}$
                            \ELSIF {$F_B(\beg_B(j)-1,k'+1) > \beg_A(i)-1$}
                                \STATE $k_1 \leftarrow k'$
                                \STATE $k^B_1 \leftarrow k'$
                                \STATE $\mathbf{break}$
                            \ENDIF
                        \ENDFOR
                    \ELSE
                        \STATE $k_1 \leftarrow 0$
                        \STATE $k^A_1 \leftarrow k$
                        \STATE $\mathbf{break}$
                    \ENDIF
                \ENDIF
            \ENDFOR
        \end{algorithmic}
        \vspace{1ex}
    \end{spacing}
  \end{algorithm}

  \begin{algorithm}[tbh]
    \begin{spacing}{0.9}
        \caption{Computing $\lcs(A[\en_A(i)+1..n],B[\en_B(j)+1..n])$ in Algorithm~\algFt}
        \label{alg:Serch4LCS_2}
        \begin{algorithmic}[1]
            \FOR {$k = k^A_2$ to $1$}
                \IF {$F_{A^R}(n-\en_A(i), k) \leq n-\en_B(j)$}
                        \STATE $k_2 \leftarrow k$
                        \STATE $k^A_2 \leftarrow k$
                        \STATE $\mathbf{break}$
                \ELSIF {$F_{A^R}(n-\en_A(i), k) > n-\en_B(j)$}
                    \IF {$F_{A^R}(n-\en_A(i), k-1) = \Undef$}
                        \STATE $k^A_2 \leftarrow k$
                        \FOR {$k' = k^B_2$ to $1$}                
                            \IF {$F_{B^R}{n-\en_B(j),k'} \leq n-\en_A(i)$}
                                \STATE $k_2 \leftarrow k'$
                                \STATE $k^B_2 \leftarrow k'$
                                \STATE $\mathbf{break}$
                            \ELSIF {$F_{B^R}{n-\en_B(j),k'-1} = \Undef$}
                                \STATE $k_2 \leftarrow 0$
                                \STATE $k^B_2 \leftarrow k'$
                                \STATE $\mathbf{break}$
                            \ENDIF
                        \ENDFOR
                    \ENDIF
                \ENDIF
            \ENDFOR
        \end{algorithmic}
        \vspace{1ex}
    \end{spacing}
  \end{algorithm}

First, our algorithm computes the sets of minimal intervals $I_A$ and $I_B$ (similar to the algorithm by Deorowicz~\cite{STRICLCS_DEOROWICZ_2012}).
Second, compute the tables $F_A$ and $F_B$ for computing LCSs of prefixes,
and the tables $F_{A^R}$ and $F_{B^R}$ for computing LCSs of suffixes (similar to the algorithm by Nakatsu et al.~\cite{DBLP:journals/acta/NakatsuKY82}).
Third, for any pairs of intervals in $I_A$ and $I_B$, compute the length of an LCS of corresponding prefixes/suffixes and
obtain a candidate of the length of an STR-IC-LCS.
As stated above, the first and the second steps are similar to the previous work.
Here, we describe our method to compute the length of an LCS of prefixes on $F_A$ and $F_B$ in the third step.
We can also compute the length of an LCS of suffixes on $F_{A^R}$ and $F_{B^R}$ 
\yynote*{modified on comment 14}{%
in a similar way.
}%


We assume that $I_A$ and $I_B$ are sorted in increasing order of the beginning positions.
Let $[\beg_A(x)..\en_A(x)]$ and $[\beg_B(y)..\en_B(y)]$ denote the $x$-th interval in $I_A$ and the $y$-th interval in $I_B$, respectively.
We process $O(n^2)$-queries in increasing order of the beginning position of the intervals in $I_A$.
For each interval $[\beg_A(x)..\en_A(x)]$ in $I_A$, we want to obtain the length of an LCS of $A[1..\beg_A(x)-1]$ and $B[1..\beg_B(1)-1]$.
For convenience, let $i_x = \beg_A(x)-1$ and $j_y = \beg_B(y)-1$.
In the rest of this section, we use a pair $(x, y)$ of integers to denote a prefix-LCS query (computing $\lcs(A[1..i_x],B[1..i_y])$).
We will find the LCS by using Observation~\ref{obs:lcs}.
Here, we describe how to compute prefix-LCS queries $(i_x,j_1), \ldots, (i_x,j_{|I_B|})$ in this order for a fixed $i_x$.

\begin{lemma} \label{lem:required-prefix-LCS}
    All required prefix-LCS values for an interval $[\beg_A(x)..\en_A(x)]$ in $I_A$ and all intervals in $I_B$
    can be computed in $O(n)$ time.
\end{lemma}

\begin{proof}
There exist two cases for each $i_x$.
Formally, (1) $F_A(1, i_x) \neq \Undef$ or (2) $F_A(1, i_x) = \Undef$.

In the first case, we scan the $i_x$-th column of $F_A$ from the top to the bottom in order to find the maximum value which is less than or equal to $j_1$.
If such a value exists in the column, then the row number $s_1$ is the length of an LCS.
After that, we are given the next prefix-LCS query $(i_x,j_2)$.
It is easy to see that $s_0 = \lcs(A[1..i_x],B[1..j_1]) \leq \lcs(A[1..i_x],B[1..j_2])$ since $j_1 < j_2$.
This implies that the next LCS value is equal to $s_0$ or that is placed in a lower row in the column.
This means that we can start to scan the column from the $s_0$-th row.
Thus we can answer all prefix-LCSs for a fixed $i_x$ in $O(n)$ time (that is linear in the size of $I_B$).

In the second case, we start to scan the column from the top $F_A(i_x-n-\ell+1, i_x)$ (the first $i_x-n-\ell$ rows are $\Undef$).
If $F_A(i_x-n-\ell+1, i_x) \leq j_1$, then the length of an LCS for the first query $(i_x,j_1)$ can be found in the table (similar to the first case)
and any other queries $(i_x,j_2), \ldots, (i_x,j_{|I_B|})$ can be also answered in a similar way.
Otherwise (if $F_A(i_x-n-\ell+1, i_x) > j_1$), the length which we want may be in the ``undefined'' domain.
Then we use the other table $F_B$.
We scan the $j_1$-th column in $F_B$ from the top to the bottom in order to find the maximum value which is less than or equal to $i_x$.
By Lemma~\ref{lem:recover-lcs}, such a value must exist in the column (if $\lcs(A[1..i_x],B[1..j_1])>0$ holds)
and the row number $s'$ is the length of an LCS.
After that, we are given the next query $(i_x,j_2)$.
If $F_A(i_x-n-\ell+1, i_x) \leq j_2$, then the length can be found in the table (similar to the first case).
Otherwise (if $F_A(i_x-n-\ell+1, i_x) > j_2$), the length must be also in the "undefined" domain.
Since such a value must exist in the $j_2$-th column in $F_B$ by Lemma~\ref{lem:recover-lcs}, we scan the column in $F_B$.
It is easy to see that $s' = \lcs(A[1..i_x],B[1..j_1]) \leq \lcs(A[1..i_x],B[1..j_2])$.
This implies that the length of an LCS that we want to find is in lower row.
Thus it is enough to scan the $j_2$-th column from the $s'$-th row to the bottom.
Then we can answer the second query $(i_x,j_2)$.
Hence we can compute all LCSs for a fixed $i_x$ in $O(n + \ell)$ time
(that is linear in the size of $I_B$ or the number of rows in the table $F_B$).

Therefore we can compute all prefix-LCSs for each interval in $I_A$ in $O(n)$ time (since $n \geq \ell$).
\end{proof}

On the other hand, we can compute all required suffix-LCS values with computing prefix-LCS values.
We want a suffix-LCS value of $A[\en_A(x)+1..n]$ and $B[\en_B(y)+1..n]~(1 \leq y \leq |I_B|)$ when we compute the length of an LCS of $A[1..\beg_A(x)-1]$ and $B[1..\beg_B(y)-1]$.
Recall that we process all intervals of $I_B$ in increasing order of the beginning positions when computing prefix-LCS values with a fixed interval of $I_A$.
This  means that we need to process all intervals of $I_B$ in ``decreasing order'' when computing suffix-LCS values with a fixed interval of $I_A$.
We can do that by using an almost similar way on $F_{A^R}$ and $F_{B^R}$.
The most significant difference is that we scan the $|A[\en_A(x)+1..n]|$-th column of $F_{A^R}$ from the $\ell$-th row to the first row.

Overall, we can obtain the length of an STR-IC-LCS in $O(n^2)$ time in total.
Also this algorithm requires space for string all minimal intervals and tables,
namely, requiring $O(n + (\ell+1)(n-\ell+1)) \subseteq O((\ell+1)(n-\ell+1))$ space in the worst case.
Finally, we can obtain Theorem~\ref{thm:str-ic-lcs}.

In addition, we can also compute an STR-IC-LCS (as a string),
if we store a pair of minimal intervals which produce the length of an STR-IC-LCS.
Namely, we can find a cell which gives the prefix-LCS value over $F_A$ or $F_B$.
Then we can obtain a prefix-LCS string by a simple backtracking (a suffix-LCS can be also obtained by backtracking on $F_{A^R}$ or $F_{B^R}$).
On the other hand, we can also use an algorithm that computes an LCS string in $O(n^2)$ time and $O(n)$ space by Hirschberg~\cite{Hirschberg75}.

\section{Faster solution (Algorithm~\algSd) for STR-IC-LCS problem}\label{sec:f_solution}

In this section, we propose a faster solution (Algorithm~\algSd) for the STR-IC-LCS problem that achieves the following:
\begin{theorem}\label{thm:str-ic-lcs_new}
  The STR-IC-LCS problem can be solved in $O(nr/\log{r}+n(n-\ell+1))$ time and $O((\ell+1)(n-\ell+1))$ space
  where $\ell$ is the length of an LCS of $A$ and $B$.
\end{theorem}

Algorithm~\algSd with $O(nr/\log{r}+n(n-\ell+1))$ running time
is built on the previous solution Algorithm~\algFt with $O(n^2)$ running time.
Thus, in this section we discuss only the differences between Algorithm~\algFt and Algorithm~\algSd.

An overview of Algorithm~\algSd is as follows:
First, we compute the sets of minimal intervals $I_A$ and $I_B$ 
by using Das et al.'s algorithm~\cite{Das_EM} that works in $O(nr/\log{r})$ time and $O(n)$ space.
Second, we compute two arrays $\theta_A$ and $\theta_B$ from $I_A$ and $I_B$ in $O(n)$ time and space, respectively.
Third, following the approach of Algorithm~\algFt,
we compute the tables $F_A$ and $F_B$ for computing LCSs of the prefixes,
and the tables $F_{A^R}$ and $F_{B^R}$ for computing LCSs of the suffixes,
in faster $O(n(n-\ell+1))$ time with $O((\ell+1)(n-\ell+1))$ space.
Finally, for some pairs of intervals in $\theta_A$ and $\theta_B$ which respectively
have intervals from $I_A$ and $I_B$,
we compute the length of an LCS of corresponding prefixes/suffixes and
obtain a candidate of the length of an STR-IC-LCS by using the faster method that works in $O((\ell+1)(n-\ell+1))$ time. 
The main idea of this step is to omit rescanning of the same element in $F_A$ and $F_B$.

In the following subsections, we explain the details of Algorithm~\algSd.
Section~\ref{subsec:theta} shows how to compute $\theta_A$ and $\theta_B$.
Section~\ref{subsec:compute-candidate}
shows the faster method of computing candidates of \emph{STR-IC-LCS}.

\subsection{Computing $\theta_A$ and $\theta_B$}\label{subsec:theta}
As stated before, $I_A$ and $I_B$ are the sets of minimal intervals 
over $A$ and $B$ which have $P$ as a subsequence, respectively.
We define two arrays $\theta_{A}$ and $\theta_{B}$:
\sinote*{modified}{%
For each $1 \leq i \leq n$,
\begin{itemize}
 \item If there is an interval $[i..j]$ in $I_A$ that begins with $i$,
then $\theta_{A}[i]$ stores this interval $[i..j]$.
These minimal intervals stored in $\theta_A$ are called \emph{interval elements} of $\theta_A$.
 \item If there is no interval in $I_A$ that begins with $i$,
then $\theta_{A}[i]$ stores the smallest index $s$
such that $s > i$ and $\theta_{A}[s]$ stores an interval element.
In case where there is no such interval to the right of $i$, then $\theta_{A}[i]$ stores null.
\end{itemize}
}%

For example,
Figure~\ref{fig:theta_ab} shows $\theta_{A}$ of the set $I_A$ of minimal intervals for string $A=\mathtt{cabcadbab}$ and pattern $\mathtt{ab}$.
The array $\theta_{B}$ is defined analogously for string $B$.
These arrays are used for our efficient algorithm in Section~\ref{subsec:compute-candidate} for computing candidates.

\begin{figure}[H]
  \centerline{\includegraphics[width=0.9\linewidth]{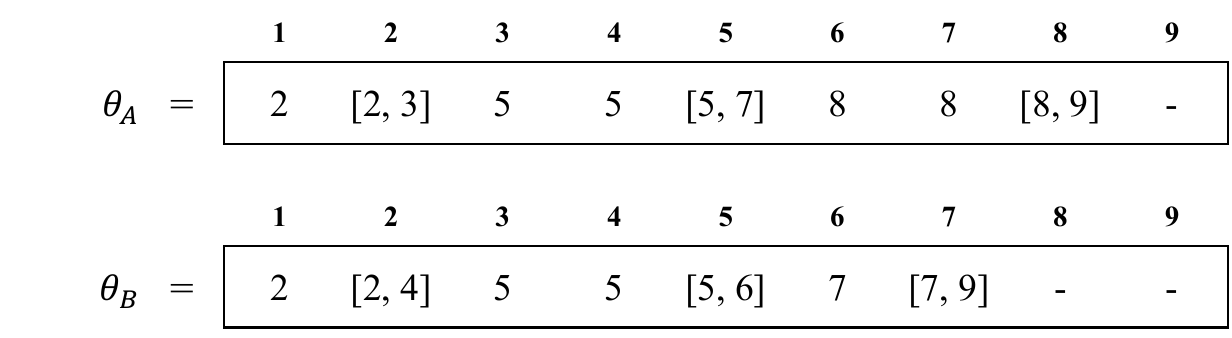}}
  \caption{$\theta_{A}$ of the set $I_A = \{[2,3], [5,7], [8,9]\}$ of minimal intervals on string $A=\mathtt{cabcadbab}$ and pattern $\mathtt{ab}$,
           and $\theta_{B}$ of the set $I_B = \{[2,4], [5,6], [7,9]\}$ of minimal intervals on string $B=\mathtt{cacbabadb}$ and pattern $\mathtt{ab}$.
  }\label{fig:theta_ab}
\end{figure}

\subsection{Computing candidates}\label{subsec:compute-candidate}

Recall that Algorithm~\algFt computes candidates 
by scanning each column of $F_A$ and some columns $F_B$ from the top to the bottom.
The total time complexity of Algorithm~\algFt is $O(n^2)$
since it depends on the number of minimal intervals which is $O(n^2)$,
in other words, the number of candidates considered by Algorithm~\algFt is $O(n^2)$.
To achieve a faster solution, Algorithm~\algSd excludes
some ``redundant'' candidates that are considered by Algorithm~\algFt
but cannot actually be a solution.
Below, we explain such candidates.

Here, let us recall the process in Algorithm~\algFt in which prefix-LCS queries
are performed in this order $(i_x,j_k), \ldots, (i_x,j_{|I_B|})$ for a fixed $i_x$ ($1 \leq k < |I_B|$).
Assume that $F_A(s_k, i_x)$ satisfying $F_A(s_k, i_x) \geq j_k$ exists.
We scan the $i_x$-th column of $F_A$ in order to find the value $F_A(s_k, i_x)$. 
After this scanning, we are given the next prefix-LCS query $(i_x,j_{k+1})$
and restart to scan the column from $F_A(s_1, i_x)$.
In this situation, if $s_k=\lcs(A[1..i_x],B[1..j_k]) = \lcs(A[1..i_x],B[1..j_{k+1}])$,
we scan only $F_A(s_k, i_x)$.
However, the candidate obtained by such scanning cannot be a solution of STR-IC-LCS
since $\lcs(A[1..i_x],B[1..j_k]) + |P| + \lcs(A[i'_x..n],B[j'_k..n]) \leq 
\lcs(A[1..i_x],B[1..j_{k+1}]) + |P| + \lcs(A[i'_x..n],B[j'_{k+1}..n])$ holds, 
where $i'_x$, $j'_k$ and $j'_{k+1}$ satisfy 
$(i_x+1, i'_x-1) \in I_A$, $(i_j+1, i'_j-1) \in I_B$ and $(i_{j+1}+1, i'_{j+1}-1) \in I_B$,
respectively.
This implies that the rescanning of the same value in $F_A$ is useless in computing a solution, 
and the symmetric arguments hold when we compute $F_B$.
We will prove Theorem~\ref{thm:str-ic-lcs_new} based on this observation.

Algorithm~\algSd performs prefix-LCS and suffixes LCS-queries
which are given from $\theta_A$ and $\theta_B$, respectively.
Thus, we describe how to perform prefix-LCS and suffix-LCS queries
by scanning the elements $\theta_A[i]$ and $\theta_B[1], \ldots, \theta_B[n]$ in this order for a fixed $i$ ($1 \leq i \leq n$).

\begin{lemma}
  We can compute all prefix/suffix LCSs with respect to $\theta_A$ and $\theta_B$
  in a total of $O((\ell+1)(n-\ell+1))$ time and space.
\end{lemma}

\begin{proof}
We use an array $W_{A^R}$ of size $n$ such that, 
\sinote*{added}{for each $1 \leq i \leq n$},
$W_{A^R}[i]$ initially stores the largest row index of the $i$-th column of 
$F_{A^R}$ that stores a non-$\infty$ value.
We also use an array $W_{B^R}$ of size $n$ that is built on $F_{B^R}$ in an analogous manner.
See Figure~\ref{fig:reverse-tables-array-w} for examples of initialized $W_{A^R}$ and $W_{B^R}$.
In our faster algorithm to follow,
we update these two arrays $W_{A^R}$ and $W_{B^R}$
while keeping the invariants that $W_{A^R}[i]$ and $W_{B^R}[i]$
respectively store the row indices of
the $i$-th columns of $F_{A^R}$ and $F_{B^R}$ from which we restart the scanning,
for every $1 \leq i \leq n$.
  
\begin{figure}[H]
      \centerline{
      \includegraphics[width=1.0\linewidth]{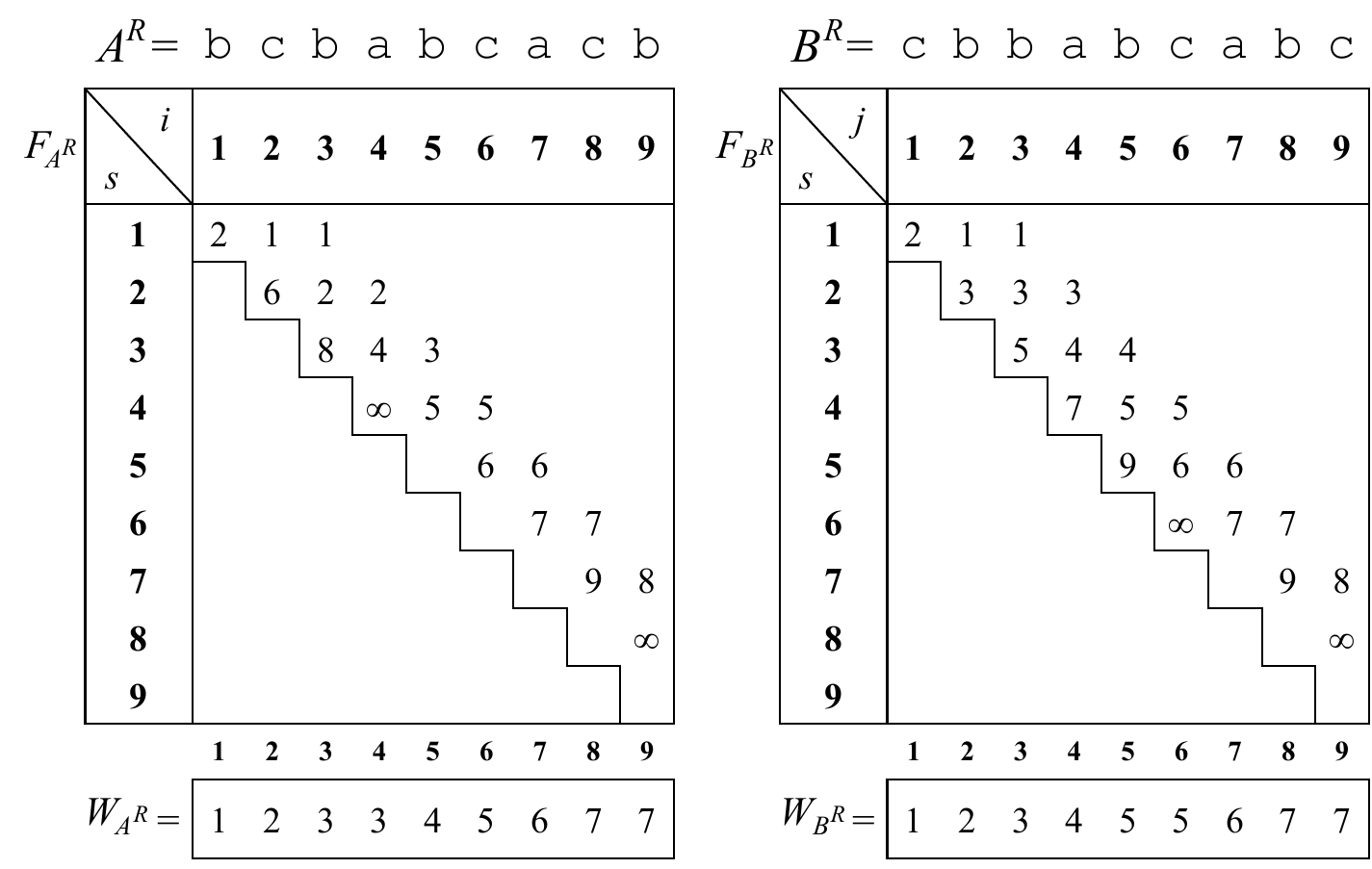}
      }
      \caption{
        Sparse tables $F_{A^R}$ and $F_{B^R}$ for $A = \mathtt{bcacbabcb}$ and $B = \mathtt{cbacbabbc}$,
        and integer arrays $W_{A^R}$ and $W_{B^R}$ in their initial states, whose $i$-th elements are initialized with the largest row indices of the $i$-th columns of $F_{A^R}$ and $F_{B^R}$ storing non-$\infty$ values, respectively.
      }\label{fig:reverse-tables-array-w}
\end{figure}

\sinote*{modified}{%
For each $1 \leq i \leq n$ such that
$\theta_A[i]$ stores an interval element, let us denote this interval by $[i..\en_{i}]$.
For a given $i$,
if $\theta_A[i]$ does not store an interval element,
then we access the next interval element in $\theta_A$ by updating the index $i$ with $i \leftarrow \theta_A[i]$.
Let $j$ be the least index in $\theta_B$
such that $\theta_B[j]$ stores an interval element.
We then proceed as described above.
}%

There exist two cases for each $i$ such that $\theta_A[i]$ stores an interval element:

\begin{itemize}
\item[(1)] When $F_A(1, i-1) \neq \Undef$:
In this case, first, we compute
the prefix-LCS length of $A[1..i-1]$ and $B[1..j-1]$.
We scan the $(i-1)$-th column of $F_A$ from the top to the bottom
in order to find the maximum value which is less than or equal to $j-1$.
If such a value exists in the column, then the row number $s$ is the length 
\yynote*{modified on comment 14}{%
of a prefix-LCS.
}%

Next, we compute
the suffix-LCS length of $A[\en_{i}+1..n]$ and $B[\en_{j}+1..n]$.
If $F_{A^R}(1, n-\en_{i}) \neq \Undef$, we scan the $(n-\en_{i})$-th column of $F_{A^R}$
by the same methods as Algorithm \algFt.
If $F_{A^R}(1, n-\en_{i}) = \Undef$,
we scan the $(n-\en_{i})$-th column of $F_{A^R}$ by the same method as Algorithm \algFt
if 
$F_{A^R}(\en_{i}-\ell, n-\en_{i})$,
and otherwise we can scan the $(n-\en_{j})$-th column of $F_{B^R}$ from the $(W_{B^R}[n-\en_{j}])$-th row 
to the top and we get the row number $s'$
since $n-\en_{i}$ is decreasing order
and $\lcs(A[\en_{i}+1+k..n],B[\en_{j}+1..n]) \geq \lcs(A[\en_{i}+1..n],B[\en_{j}+1..n])$ for $1 \leq k \leq n-(\en_{i}+1)$.
Then, in the case that we scan $F_{B^R}$, we store $s'$ at $W_{B^R}[n-\en_{j}]$.

After that,
  we update $j$ with the beginning position of the next interval element,
  which is either directly stored at $\theta_B[F_A(s+1, i-1)+1]$
  or pointed by a pointer stored at $\theta_B[F_A(s+1, i-1)+1]$.
This operation allows us to skip all interval elements which can cause rescanning $F_A(s, i-1)$.
This way, we can start scanning from $F_A(s+1, i-1)$ in order to avoid rescanning $F_A(s, i-1)$.
Thus we can answer all prefix-LCSs for a fixed $i$ in time linear in the size of the $i-1$-th row of $F_A$,
meaning that we can compute the prefix-LCSs for all $i$ in time linear in the size of $F_A$, which is $O((\ell+1)(n-\ell+1))$.
We remark that, in other words, the number of candidates has been reduced to the size of $F_A$.

In addition, for obtaining the next suffix-LCS length of $A[\en_{i}+1..n]$ and $B[\en_{j}+1..n]$,
  we can use the same method as in the case of scanning only $F_{A^R}$ with Algorithm~\algFt.
In this case,
we scan the $(n-\en_{j})$-th column of $F_{B^R}$ from the $(\min\{s', W_{B^R}[n-\en_{j}]\})$-th row to the top and we obtain the row number $t$.
Then, we store $t$ at $W_{B^R}[n-\en_{j}]$.

Overall,
since the total number of rescanned cells in $F_{A^R}$ and $F_{B^R}$ is linear in the size of $F_A$,
and since the total number of scanned cells except for the aforementioned rescanned cells is linear in the total size of $F_{A^R}$ and $F_{B^R}$,
we can answer the suffix-LCSs for all $i$ in
time linear in the total size of $F_{A^R}$, $F_{B^R}$ and $F_A$,
which is $O((\ell+1)(n-\ell+1))$.

\item[(2)] When $F_A(i-1,1) = \Undef$:
In this second case,
\sinote*{reworded}{%
  we compute candidates (i.e. the sums of the lengths of pairs of a prefix-LCS and suffix-LCS)
  as follows:
}%
We start scanning the column from the top $F_A(i-n-\ell, i-1)$.
If $F_A(i-n-\ell, i-1) \leq j-1$,
then the length 
\yynote*{modified on comment 15}{%
of a prefix-LCS
}%
for the first query $(i-1,j-1)$ can be found in the table (similar to Case (1))
and next queries about interval elements can also be answered in a similar way.
Otherwise (if $F_A(i-n-\ell, i-1) > j-1$),
the length which we want may be in the ``undefined'' domain.
Then we use the other table $F_B$.
We scan the $(j-1)$-th column in $F_B$ from the top to the bottom 
in order to find the maximum value which is less than or equal to $i-1$.
Then, the row number $s$ is the length of a
prefix-LCS of $A[1..i-1]$ and $B[1..j-1]$.
In addition, we compute
the suffix-LCS length of $A[\en_{i}+1..n]$ and $B[\en_{j}+1..n]$
by the same method as Case (1).

After that, we compute
the prefix-LCS of $A[1..i'-1]$ and $B[1..j-1]$, and the suffix-LCS of $A[\en_{i'}+1..n]$ and $B[\en_{j}+1..n]$ for
\sinote*{reworded}{%
the fixed $j$ and the next index $i'$~($i < i' \leq n$) in a similar way,
where $i'$ is the index of the next interval element that is stored after $i$ in $\theta_A$.
}%
We obtain the next interval element $[i'..\en_{i'}]$ or an index to that element 
from $\theta_A[F_B(s+1, j-1)+1]$.
Here again, we can start scanning from $F_B(s+1, j-1)$ in order to avoid rescanning $F_B(s, j-1)$
for computing the prefix-LCS of $A[i'-1..n]$ and $B[1..j-1]$.
In addition, we compute
the suffix-LCS of $A[\en_{i'}+1..n]$ and $B[\en_{j}+1..n]$
in a similar way to Case (1).
The difference is that we scan $F_{A^R}$ and we use $W_{A^R}$ instead of $W_{B^R}$.
We repeat this method up to the end of $\theta_A$.

We then delete the interval element $\theta_B[j] = [j, \en_{j}]$ and 
$\theta_B[1]$ is assigned to $j+1$, 
since we have completed computation for the interval $[j, \en_{j}]$.
With these updates done, we are able to access the interval elements
by accessing at most two elements on $\theta_B$.
Next, we search $\theta_B[j+1]$ for the next interval and repeat the above operations.

Overall, we can answer all prefix-LCSs in $O((\ell+1)(n-\ell+1))$ for all $i$ (that is linear in the sum of the sizes of $F_A$ and $F_B$).
In other words, the number of candidates has been reduced to the total size of $F_A$ and $F_B$.
Since the total number of rescanned cells in $F_{A^R}$ and $F_{B^R}$
is linear in the total size of $F_A$ and $F_B$,
and since the total number of scanned cells except for the aforementioned
rescanned cells is linear in the total size of $F_{A^R}$ and $F_{B^R}$,
we conclude that the suffix-LCSs for all $i$ can be computed
in time linear in the total sizes of $F_A$, $F_B$, $F_{A^R}$, and $F_{B^R}$,
which is $O((\ell+1)(n-\ell+1))$. 
\end{itemize}

Finally, we can compute all prefix/suffix-LCSs in $O((\ell+1)(n-\ell+1))$ total time.
\end{proof}

Below we present some concrete examples on how our Algorithm~\algSd operates efficiently.

\begin{figure}[H]
  \centerline{
  \includegraphics[width=1.0\linewidth]{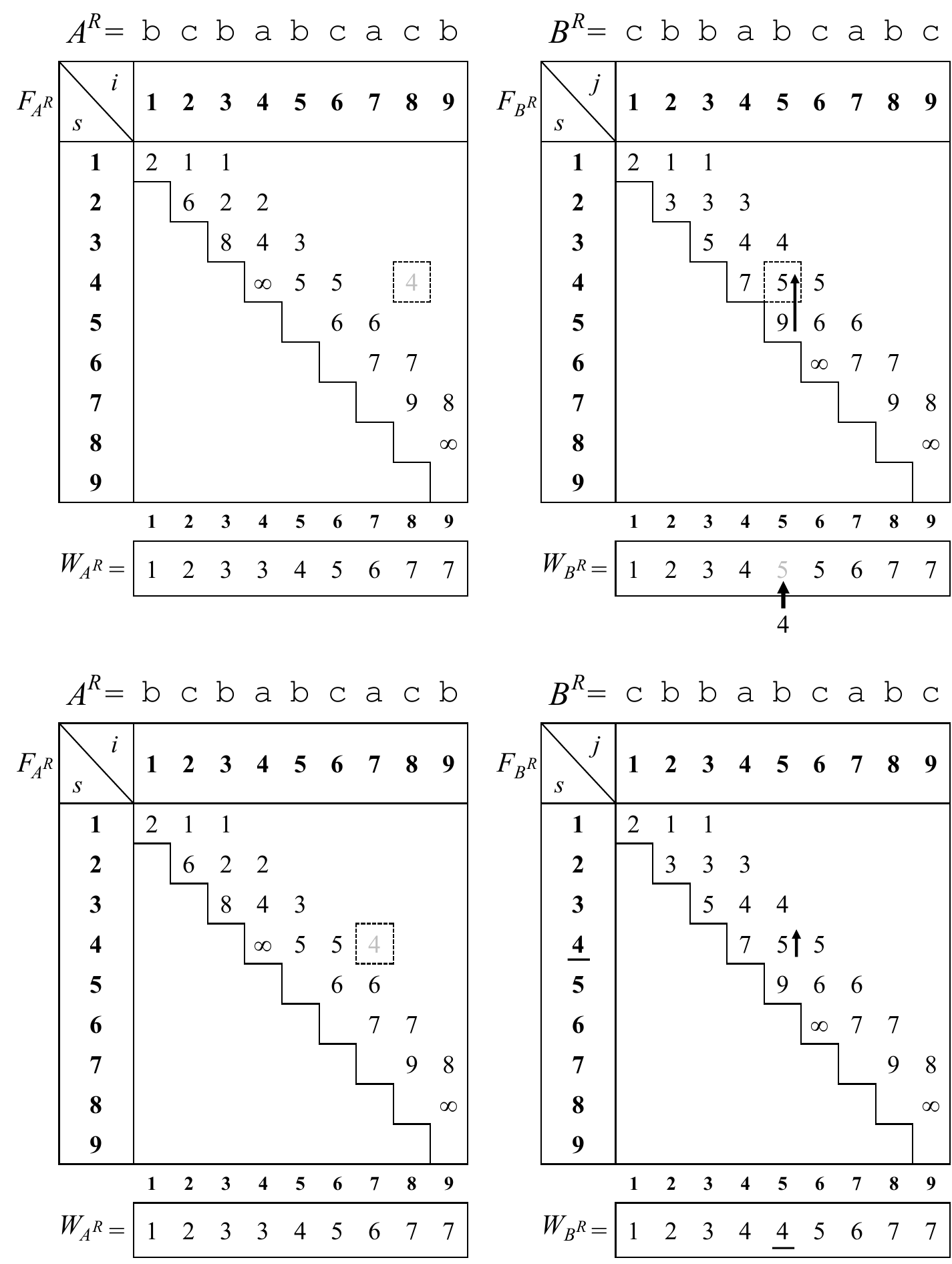}
  }
  \caption{
    An example of scanning $F_{A^R}, F_{B^R}$ and updating $W_{B^R}$ in prefix/suffix-LCS computation in Algorithm~\algSd.
    After the computation in the upper table, we can skip scanning $F_{B^R}(5, 5)$ in the lower table.
  }\label{fig:reverse-tables-array-w_2}
\end{figure}

\begin{example}
  Here we show an example of scanning $F_{A^R}, F_{B^R}$ and updating $W_{B^R}$ in prefix/suffix-LCS computation in Algorithm~\algSd.
  See in Figure~\ref{fig:reverse-tables-array-w_2} for illustrations.
  
  Consider searching for $s$ with $n-\en_{i}=8$ and $n-\en_{j}=5$ by scanning from $F_{A^R}$ and $F_{B^R}$.
  Here we cannot obtain $s$ from $F_{A^R}$. Thus, we start scanning from $F_{B^R}(W_{B^R}[5], 5) = F_{B^R}(5, 5)$ toward the top of the column, and we obtain $s=4$ from $F_{B^R}(4, 5) = 5 \leq 8$.
  Then, we update $W_{B^R}[5] \leftarrow s = 4$.

  Next, we search for $s$ with $n-\en_{i}=7$ and $n-\en_{j}=5$ (since $n-\en_{i}$ is in decreasing order).
  Likewise, we cannot obtain $s$ from $F_{A^R}$. Then we can restart scanning from $F_{B^R}(W_{B^R}[5], 5) = F_{B^R}(4, 5)$ toward the top of the column for $s$. This permits us to avoid rescanning $F_{B^R}(5, 5)$.
\end{example}

%
%

Overall, we can obtain the length of an STR-IC-LCS in $O(nr/\log{r} + n(n-\ell+1))$ time in total.
Also this algorithm requires space for storing all minimal intervals, tables, and arrays,
namely, requiring $O(n + (\ell+1)(n-\ell+1)) = O((\ell+1)(n-\ell+1))$ space in the worst case.
We have thus proven Theorem~\ref{thm:str-ic-lcs_new}.


In addition, we can also compute an STR-IC-LCS (as a string) 
by storing a pair of minimal intervals that correspond to an STR-IC-LCS,
and then using the same method as Algorithm~\algFt.

\section{Alternative solution (Algorithm~\algTd) for STR-IC-LCS problem}\label{sec:alternative_solution}

In this section, we present our third solution Algorithm~\algTd to STR-IC-LCS.
In the sequel, $\ell'$ denotes the solution length of the STR-IC-LCS for input strings $A$, $B$, and $P$.

We first compute the minimal intervals in $A$ and $B$ where $P$ occurs as subsequence, using $O(nr / \log r)$ time as in Section~\ref{sec:f_solution}.
Let $\max \beg_A$ and $\max \beg_B$ respectively denote the beginning positions of
the \emph{rightmost} minimal intervals in $A$ and $B$ where $P$ occurs as subsequence.

The following observation and lemma are a key to our third solution Algorithm~\algTd.

\begin{observation} \label{obs:redundant_positions}
  Prefix-LCS queries are never performed in Algorithm~\algSd
  for any pair $(i,j)$ of positions such that $\max \beg_A \leq i \leq n$ and $\max \beg_B \leq j \leq n$.
\end{observation}

\begin{lemma} \label{lem:key_third_solution}
  For any $1 \leq i < \max \beg_A$ and $1 \leq j < \max \beg_B$,
  $\lcs(A[1..i], B[1..j]) \leq \ell'$.
\end{lemma}

\begin{proof}
  Assume on the contrary that there exists a pair $(i,j)$
  of positions in $A$ and $B$ such that $1 \leq i < \max \beg_A$, $1 \leq j < \max \beg_B$, and $l = \lcs(A[1..i], B[1..j]) > \ell'$.
  Then, there must exist an STR-IC-LCS of $A$, $B$, and $P$ of length at least $l+r > \ell' + r \geq \ell'$.
  However, this contradicts that $\ell'$ is the STR-IC-LCS of $A$, $B$, and $P$.
\end{proof}

Due to Observation~\ref{obs:redundant_positions}, it suffices for us to work on the prefixes $\hat{A} = A[1..\max \beg_A-1]$ and $\hat{B} = B[1..\max \beg_B-1]$, namely, we build tables $F_{\hat{A}}$ and $F_{\hat{B}}$ and the other auxiliary tables for $\hat{A}$ and $\hat{B}$.
Then, it follows from Lemma~\ref{lem:key_third_solution} that
the sparse tables $F_{\hat{A}}$ and $F_{\hat{B}}$ (and the other auxiliary tables)
require $O((\ell'+1)(n-\ell'+1))$ space,
and can be computed in $O(n(n-\ell'+1))$ time,
by applying the method described in Section~\ref{sec:f_solution} to $A'$ and $B'$.
We can use the same techniques for suffix-LCS queries.

Overall, we obtain the following:
\begin{theorem}\label{thm:str-ic-lcs_alternative}
  The STR-IC-LCS problem can be solved in $O(nr/\log{r}+n(n-\ell'+1))$ time and $O((\ell'+1)(n-\ell'+1))$ space,
  where $\ell'$ is the length of an STR-IC-LCS of $A$, $B$, and $P$.
\end{theorem}

The merit of Algorithm~\algTd when compared to Algorithm~\algSd is summarized in the following lemma:
\begin{lemma}
  There exist strings of length $n$
  for which Algorithm~\algTd uses only $O(n)$ space,
  while Algorithm~\algSd needs $\Theta(n^2)$ space.
\end{lemma}

\begin{proof}
  Consider two strings $A = a^{n-1}$ and $B = a^{\frac{n}{2}}b^{\frac{n}{2}-1}$
  of length $n-1$ each
  such that $\ell = \lcs(A, B) = n/2$.
  Then, consider strings $A' = a^i{c}a^{n-i-1}$ and $B' = a^{\frac{n}{2}-i}ca^{i}b^{\frac{n}{2}-1}$ of length $n$ each, where $c \in \Sigma \setminus \{a, b\}$, and $i$ is a constant.
  Then, the solution length $\ell'$ of the STR-IC-LCS of $A'$, $B'$, and $P = c$ is $2i+1= O(1)$.
  Thus, Algorithm~\algTd uses $O((\ell'+1)(n-\ell'+1)) = O(n)$ space,
  while Algorithm~\algSd uses $O((\ell+1)(n-\ell+1)) = \Theta(n^2)$ space
  for $A'$, $B'$, and $P$.
\end{proof}

\section{Conclusions}

This paper proposed three space-efficient algorithms that find
an STR-IC-LCS of two given strings $A$ and $B$ of length $n$
with constrained pattern $P$ of length $r$.

Our first solution, Algorithm~\algFt, works 
in $O(n^2)$ time with $O((\ell+1)(n-\ell+1))$ working space,
where $\ell$ is the length of an LCS of $A$ and $B$.
This method improves on the space requirements of
the algorithm by Deorowicz~\cite{STRICLCS_DEOROWICZ_2012}
that uses $\Theta(n^2)$ space, irrespective of the value of $\ell$.

Our second solution, Algorithm~\algSd,
runs in faster $O(nr/\log{r}+n(n-\ell+1))$ time with the same $O((\ell+1)(n-\ell+1))$ working space.
We have achieved this improved $O(nr/\log{r}+n(n-\ell+1))$-time complexity by carefully avoiding redundant scans in the dynamic programming tables.

Our third solution, Algorithm~\algTd,
runs in $O(nr/\log{r}+n(n-\ell'+1))$ time with $O((\ell'+1)(n-\ell'+1))$ working space, where $\ell'$ is the STR-IC-LCS length.

We note that all of our proposed algorithms are based on Nakatsu et al.'s
algorithm~\cite{DBLP:journals/acta/NakatsuKY82} for finding (standard) LCS that runs in $O(n(n-\ell+1))$ time
with $O((\ell+1)(n-\ell+1))$ working space.
The only overhead in our faster solution Algorithm~\algSd is the $nr/\log{r}$ additive factor for finding minimal intervals where $P$ occur by using Das et al.'s method~\cite{Das_EM}.
Bille et al.~\cite{BilleGMSW22} showed that
there is no strongly sub-quadratic $O((nr)^{1-\epsilon})$ time
algorithm for finding such minimal intervals unless the famous SETH fails.
Thus, as long as computing such minimal intervals is involved,
it might be difficult to drastically improve our $O(nr/\log{r}+n(n-\ell+1))$ time complexity
for STR-IC-LCS.


\section*{Acknowledgments}
This work was supported by JSPS KAKENHI Grant Numbers JP21K17705 (YN), JP20H05964, JP22H03551 (SI), JP20H04141 (HB), and by JST PRESTO Grant Number JPMJPR1922 (SI).

\bibliographystyle{abbrv}
\bibliography{ref}
\end{document}